\documentclass[sigconf, nonacm]{acmart}

\usepackage{float}
\usepackage{amsmath, amsthm}
\usepackage{stmaryrd}

\usepackage{amssymb}
\usepackage{bbm}
\usepackage{mathtools}
\usepackage{algorithm}
\usepackage{xfrac}
\usepackage{bm}
\usepackage{xcolor}
\usepackage{comment}
\usepackage{dsfont}
\usepackage{booktabs}
\usepackage{cprotect}
\usepackage{todonotes}
\usepackage{pdfpages}
\usepackage{pseudo}
\usepackage{tikz-cd}
\usepackage{adjustbox}

\definecolor{mygreen}{HTML}{508432}
\definecolor{myyellow}{HTML}{c25d0f}
\definecolor{myred}{HTML}{bf1716}

\usepackage[noend]{algorithmic}

\pseudoset{
  compact,
  label=\scriptsize\sffamily\color{gray}\arabic*,
  ctfont=\sffamily\color{gray},
  ct-left=\ctfont{\# },
  ct-right=,
  indent-length=3ex,
}

\usepackage{tcolorbox} 
\tcbuselibrary{skins, theorems}

\usepackage{aliascnt}
\usepackage[capitalise]{cleveref}

\def\xx{\mathbf{x}}
\def\xxi{\mathbf{\xi}}
\def\tt{\mathbf{t}}

\def\dd{\bm{\partial}}
\def\aalpha{{\bm\alpha}}
\def\bbeta{{\bm\beta}}

\def\ggamma{{\bm\gamma}}
\def\xxi{{\bm\xi}}
\def\zzeta{{\bm\zeta}}

\DeclareMathOperator{\lm}{lm}

\DeclareMathOperator{\ann}{ann}
\DeclareMathOperator{\lex}{lex}

\DeclareMathOperator{\cl}{Cl}

\DeclareMathOperator{\ini}{in}
\DeclareMathOperator{\grevlex}{grevlex}
\DeclareMathOperator{\sing}{Sing}
\DeclareMathOperator{\gens}{gens}
\DeclareMathOperator{\charv}{char}

\DeclareMathOperator{\Ai}{Ai}
\DeclareMathOperator{\J}{J}
\DeclareMathOperator{\I}{I}
\DeclareMathOperator{\K}{K}
\DeclareMathOperator{\Y}{Y}

\newcommand{\bK}{\mathbb K}
\newcommand{\bN}{\mathbb N}
\newcommand{\bQ}{\mathbb Q}

\newcommand{\bF}{\mathbb F}

\newcommand{\cF}{\mathcal F}

\newtcbtheorem{alg}{Algorithm}{label type=algorithm,pseudo/booktabs,before skip=20pt,  
    after skip=20pt}{algo}

\newenvironment{algovals}{\begin{list}{\scriptsize$\bullet$}{\setlength{\topsep}{0pt}\setlength{\parskip}{0pt}\setlength{\itemsep}{0pt}\setlength{\parsep}{0pt}}}{\end{list}}

\theoremstyle{plain}
\newtheorem{theorem}{Theorem}
\newaliascnt{proposition}{theorem}

\aliascntresetthe{proposition}
\newaliascnt{corollary}{theorem}
\newtheorem{corollary}[corollary]{Corollary}
\aliascntresetthe{corollary}
\newaliascnt{lemma}{theorem}
\newtheorem{lemma}[lemma]{Lemma}
\aliascntresetthe{lemma}
\newaliascnt{hypothesis}{theorem}

\aliascntresetthe{hypothesis}

\theoremstyle{definition}
\newaliascnt{definition}{theorem}
\newtheorem{definition}[definition]{Definition}
\aliascntresetthe{definition}
\newaliascnt{characterization}{theorem}

\aliascntresetthe{characterization}
\newaliascnt{notation}{theorem}

\aliascntresetthe{notation}

\theoremstyle{remark}
\newaliascnt{remark}{theorem}
\newtheorem{remark}[remark]{Remark}
\aliascntresetthe{remark}
\newaliascnt{example}{theorem}
\newtheorem{example}[example]{Example}
\aliascntresetthe{example}

\AtBeginDocument{%
  \crefname{theorem}{Theorem}{Theorems}%
  \crefname{proposition}{Proposition}{Propositions}%
  \crefname{corollary}{Corollary}{Corollaries}%
  \crefname{lemma}{Lemma}{Lemmas}%
  \crefname{hypothesis}{Hypothesis}{Hypotheses}%
  \crefname{definition}{Definition}{Definitions}%
  \crefname{characterization}{Characterization}{Characterizations}%
  \crefname{notation}{Notation}{Notations}%
  \crefname{remark}{Remark}{Remarks}%
  \crefname{example}{Example}{Examples}%
}

\hypersetup{
  plainpages=false,
  colorlinks=true,
  allcolors=blue!80!black,
  hypertexnames=false,
  unicode=true
}

\begin{document}

\title{Computing a holonomic submodule of the partial Weyl closure}

\author{Hadrien Brochet}
\authornote{This work was carried out while the author was at Inria, France, and was supported by the European Research 
Council under the European Union's Horizon Europe research and innovation programme, 
grant agreement 101040794 (10000 DIGITS).}
\affiliation{
  \institution{Max-Planck-Institut für Physik}
  \country{Germany}
}
\email{hbrochet@mpp.mpg.de}

\begin{abstract}
\noindent
The Weyl closure is a basic operation in algebraic analysis: it converts a system of differential operators with rational coefficients into an equivalent system with polynomial coefficients. 
In addition to encoding finer information on the singularities of the system, it serves as a preparatory step for many algorithms in symbolic integration.
A new algorithm is introduced to compute a holonomic submodule of the partial Weyl closure of a finite-rank module, where the closure is taken with respect to a subset of the variables.
The method relies on a non-commutative analogue of Rabinowitsch's trick.
The algorithm is implemented in the \textsc{Julia} package \textsc{MultivariateCreativeTelescoping.jl}
and shows substantial speedups over existing exact Weyl closure algorithms in \textsc{Singular} and \textsc{Macaulay2}.

\end{abstract}

\maketitle

\section{Introduction}

Algebraic analysis is the area of mathematics that studies systems of linear partial differential equations (PDEs)
with polynomial coefficients from an algebraic point of view. 
In this context, such systems are represented by a left ideal in the $n$th {Weyl algebra}
\begin{equation}
W_\xx = K[x_1,\dots,x_n]\langle \partial_1,\dots,\partial_n\rangle.
\end{equation} 
One important application of algebraic analysis is the study of parametric integrals through the construction of systems of PDEs that they satisfy.
This has been shown to be always possible when the integrand $f$ is {holonomic}~\cite{kashiwaraHolonomicSystemsLinear1978},
meaning that it satisfies a “maximally overdetermined" system of PDEs.
Algorithms for performing this computation exist~\cite{takayamaAlgorithmConstructingIntegral1990,oakuAlgorithmsDmodulesRestriction2001, brochetFasterMultivariateIntegration2025}.
However, these algorithms require as input generators of the annihilator of the integrand in $W_\xx$:
 \begin{equation}
     \ann_{W_\xx}(f) =\{ P\in W_\xx \mid P\cdot f = 0 \}. 
    \end{equation}
A common approach to compute $\ann_{W_\xx}(f)$ is to first compute 
 the annihilator $\ann_{W_\xx(\xx)}(f)$ of $f$ in the rational Weyl algebra
 \begin{equation}
 W_\xx(\xx) = K(x_1,\dots,x_n)\langle \partial_i,\dots,\partial_n\rangle
 \end{equation}
and then recover $\ann_{W_\xx}(f)$ via the intersection
\begin{equation} 
    \ann_{W_\xx}(f) = \ann_{W_\xx(\xx)}(f) \cap W_\xx. 
\end{equation}
This intersection computation is a particular instance of the {Weyl closure} defined as follows. 
Let $S \subset W_\xx^r$ be a left submodule for some integer $r$.  
The {Weyl closure} of $S$, denoted by $\cl(S)$, is the left $W_\xx$-module
\begin{equation}
    \cl(S) = W_\xx(\xx) S \,\cap\, W_\xx^r .
\end{equation}
Although an algorithm for computing Weyl closures exists~\cite{tsaiAlgorithmsAlgebraicAnalysis2000}, 
its practical complexity is often a bottleneck in symbolic integration, 
motivating the search for efficient approximation approaches.

A recent integration algorithm~\cite{brochetFasterMultivariateIntegration2025} 
assumes that the parameters $\tt = (t_1,\dots,t_m)$ of the integrals are rational variables,  while the  
integration variables $\xx$ are polynomial variables, and takes as input generators of 
$\ann_{W_{t,\xx}(\tt)}(f)$ where 
$W_{t,\xx}(\tt)=\bK(\tt)[\xx]\langle \partial_\tt,\partial_\xx\rangle$.
This motivates extending the definition of Weyl closure to partial Weyl closure.
Let $S\subseteq W_{\tt,\xx}(\tt)^r$ be a left submodule.
The partial Weyl closure of $S$ with respect to $\xx$ is
	  \begin{equation} 
	    \cl_\xx(S) = W_{\tt,\xx}(\tt,\xx)S\cap W_{\tt,\xx}(\tt)^r,
	 \end{equation}
where $W_{\tt,\xx}(\tt,\xx) = \bK(\tt,\xx)\langle \partial_\tt, \partial_\xx \rangle$.

 \paragraph{Related work}
 The problem of computing generators of $\cl(S)$ given $S\subset W_\xx^r$ was first introduced by Chyzak and Salvy~\cite{chyzakNoncommutativeEliminationOre1998} 
for the case $r=1$, under the name of extension/contraction problem. 
This problem was later solved by Tsai for general $r$ in his PhD thesis~\cite{tsaiWeylClosureLinear2000, tsaiAlgorithmsAlgebraicAnalysis2000}. 
Tsai showed~\cite[Theorem 2.2.1]{tsaiAlgorithmsAlgebraicAnalysis2000} that the Weyl closure of $S\subset W_\xx^r$ can be computed by inverting a single polynomial $p$ vanishing on the singular locus of $S$:
 \begin{equation} 
    \cl(S) = W_\xx[1/p]S \cap W_\xx^r. 
\end{equation} 
This yields an algorithm when combined with the localization algorithm of Oaku, Takayama, and Walther~\cite{oakuLocalizationAlgorithmDmodules2000} 
and Oaku's algorithm for computing the singular locus of $S$~\cite{oakuComputationCharacteristicVariety1994}.

The Weyl closure is also related to desingularization.
For $n=1$ the singularities of an operator $L\in W_x$ are the roots of its leading coefficient, 
and a desingularization of $L$ at a singularity $\alpha$ is an operator $L'\in W_x$ for which there exist $M\in W_x(x)$ such that $L' = ML$ and $L'$ does not have a singularity at $\alpha$. 
This notion was extended to systems of PDEs by Chen, Kauers, Li and Zhang~\cite{chenApparentSingularitiesDfinite2019}. 
The Weyl closure corresponds to the smallest ideal containing all desingularizations of the initial system.
Algorithms for the desingularization of Ore and differential operators in the univariate case have been studied extensively; see, for instance,~\cite{barkatouRemovingApparentSingularities2015, chenDesingularizationOreOperators2016, inceOrdinaryDifferentialEquations1926, zhangContractionOreIdeals2016}.

The computation of annihilators of holonomic functions in the rational Weyl algebra $W_\xx(\xx)$ is already well understood. 
Systems of differential equations are known for basic functions and algorithmic closure by sum, product and more~\cite{comtetCalculPratiqueCoefficients1964, takayamaApproachZeroRecognition1992} 
exist for deducing systems for more complicated ones. 
These algorithms are implemented in \textsc{Maple}~\cite{chyzakMgfunMaplePackage}, \textsc{Mathematica}~\cite{koutschanHolonomicFunctionsMathematica2014}, and \textsc{Sage}~\cite{kauersOrePolynomialsSage2015}.
Although similar closure properties for sum and product exist in the Weyl algebra~\cite{oakuAlgorithmsDmodulesRestriction2001} (see also~\cite[Sec.3]{oakuAlgorithmsIntegralsHolonomic2013}), 
their computational efficiency seems unclear. For example, the closure by product requires doubling the number of variables and eliminate half of them by Gröbner basis computation, which is known to be costly.

\paragraph{Contributions}
The notion of holonomicity is extended to $W_{\tt,\xx}(\tt)$-modules
and a new algorithm is proposed for computing a holonomic submodule of the partial Weyl closure. 
That is, given a finite rank module $S\subset W_{\tt,\xx}(\tt)^r$, the algorithm computes a module $S'$ such that $S'\subseteq\cl(S)$ 
and $W_{\tt,\xx}(\tt)^r/S'$ is holonomic.
The algorithm proceeds by iteratively constructing a sequence of modules 
\begin{equation} 
    S_1 \subseteq S_2 \subseteq \cdots \subseteq \cl(S) 
\end{equation} 
with the property that there exists $\ell\in\bN$ such that $S_\ell = \cl(S)$. 
The condition $S_i = \cl(S)$ could, in principle, be verified using an extension of the $b$-function associated with a weight vector~\cite[Chapter 5.1]{saitoGrobnerDeformationsHypergeometric2000}.
However, computing this function is known to be expensive, so a weaker condition is used instead, yielding an approximation algorithm. 
One possible choice is to stop at the smallest $i\in\bN$ such that $W_{\tt,\xx}(\tt)^r/S_i$ is holonomic. 
Note that, in practice, obtaining only a holonomic approximation does not seem to be an issue for symbolic integration;
closure algorithms under sums and products, both in the polynomial and rational Weyl algebra, are already not guaranteed to return the full annihilator.

The algorithm is based on a generalization of Rabinowitsch's trick to D-modules, 
which consists in introducing a new variable $T$ thought of as $1/p$ for some $p\in\bK(\tt)[\xx]$,
together with the additional equation $pT - 1 = 0$.
The commutation rules between $T$ and $\partial_\xx,\partial_\tt$ are not well suited for computing Gröbner basis 
in $W_{\xx,\tt}(\tt)[T]^r$ for an order eliminating $T$. 
This problem is tackled by computing Gröbner bases in $W_{\tt,\xx}(\tt)[T]^r$ seen as an infinite-rank $W_{\tt,\xx}(\tt)$-module~\cite{takayamaAlgorithmConstructingIntegral1990}. 
Then, saturations of $S$ with respect to $p$ are computed iteratively until the stopping condition is reached.
Lastly, the algorithm has been implemented in the \textsc{Julia} package MultivariateCreativeTelescoping.jl~\cite{brochetJuliaPackageMultivariateCreativeTelescoping2025}. 
It compares very favorably to existing implementations.

Compared with Tsai's approach, the present algorithm avoids several computationally expensive steps.
Tsai's algorithm first computes generators of the $W_\xx$-module $W_\xx[1/p]S$ using the localization algorithm
of Oaku, Takayama, and Walther~\cite{oakuLocalizationAlgorithmDmodules2000}, and then recovers the Weyl closure
by computing the intersection $W_\xx[1/p]S \cap W_\xx^r$.
The localization step requires the computation of a $b$-function associated with a weight vector of a suitable ideal,
which is known to be difficult in general.
Moreover, simplifications with respect to the polynomial $p$ are postponed until the final intersection step.
In particular, when $S$ is an ideal annihilating a function $f$, the localization algorithm returns an ideal $S'$
together with an integer $\ell$ such that $S'$ annihilates $f/p^\ell$, which typically leads to a significant blow-up
in the size of expressions compared to $S$ and $\cl(S)$.
This issue is avoided in the present algorithm by performing simplifications with respect to $p$ as early as possible
through iterative saturation.

\section{Preliminaries}
\subsection{Weyl algebras and holonomicity}
\paragraph{Polynomial and rational Weyl algebra} 
Let $\bK$ be a field of characteristic zero. Let $W_{\xx}$ denote the {$n$th Weyl algebra}
$\bK[\xx]\langle \dd_\xx \rangle$ with generators $\xx=(x_1,\dots,x_n)$ and
$\dd_\xx=(\partial_1,\dots,\partial_n)$, subject to the relations
\[
\partial_i x_i = x_i \partial_i + 1,
\quad
x_i x_j = x_j x_i,
\quad
\partial_i \partial_j = \partial_j \partial_i,
\quad
x_i \partial_j = \partial_j x_i \quad (i \neq j).
\]
The {rational Weyl algebra} $W_\xx(\xx)=\bK(\xx)\langle \dd_\xx \rangle$ is defined similarly,
with the commutation rules
\[
\partial_i \partial_j = \partial_j \partial_i,
\quad
\partial_i r(\xx)
=
r(\xx)\partial_i + \frac{\partial r(\xx)}{\partial x_i},
\]
for all $r(\xx)\in\bK(\xx)$ and $i\in\{1,\dots,n\}$.
More generally, we consider the mixed Weyl algebra $W_{\tt,\xx}(\tt)$,
where $\tt$ are rational variables and $\xx$ are polynomial variables.

Any element $P \in W_{\tt,\xx}(\tt)$ admits a unique decomposition
\[
P = \sum_{\aalpha,\bbeta,\ggamma} a_{\aalpha,\bbeta,\ggamma}(\tt)\,\xx^\aalpha \dd_\xx^\bbeta \dd_\tt^\ggamma,
\]
where the sum is finite and
$a_{\aalpha,\bbeta,\ggamma}(\tt)\in\bK(\tt)$.

The degree of $P$ is defined by
\[
\deg(P)=\max\{|\aalpha| + |\bbeta| + |\ggamma|\mid a_{\aalpha,\bbeta,\ggamma}(\tt)\neq 0\}.
\]
This definition extends naturally to free left modules.
Throughout this article, all modules are left modules, and the adjective ``left'' is omitted when no confusion may arise.
A submodule $S \subset W_{\tt,\xx}(\tt)^r$ is said to be of {finite rank}
if the vector space
\[
W_{\tt,\xx}(\tt,\xx)^r \big/ \big( W_{\tt,\xx}(\tt)\,S \big)
\]
is finite-dimensional over $\bK(\tt)$.
A module is finite-rank if and only if its Weyl closure is holonomic
\cite{takayamaApproachZeroRecognition1992}.

\paragraph{Holonomic $W_\xx$-modules.}
Let $S$ be a submodule of $W_\xx^r$, and set $M = W_\xx^r / S$.
The {Bernstein filtration}~\cite{bernsteinModulesRingDifferential1971}
of the algebra $W_\xx$ is the sequence of $\bK$-vector spaces $(\cF_m)_{m\geq 0}$
defined by
\[
\cF_m = \{ P \in W_\xx \mid \deg(P) \leq m \}.
\]
A filtration of $M$ adapted to $(\cF_m)_{m\geq 0}$ is given by the sequence
of $\bK$-linear subspaces $\Phi_m \subseteq M$ defined by
\[
\Phi_m
=
\{ P\cdot e_i + S \mid P \in \cF_m,\ 1 \leq i \leq r \},
\]
where $P\cdot e_i + S$ denotes the image of $P\cdot e_i$ in the quotient
$W_\xx^r/S$.
These filtrations are compatible with the algebra and module structures:
$\cF_m \cF_{m'} \subseteq \cF_{m+m'}$ and
$\cF_m \Phi_{m'} \subseteq \Phi_{m+m'}$.

There exists a unique polynomial $p \in \bK[m]$, called the Hilbert polynomial of $M$,
such that $\dim_\bK(\Phi_k) = p(k)$ for all sufficiently large $k$.
The {dimension} of $M$ is defined as the degree $d$ of $p$.
The integer~$d$ clearly lies between $0$ and~$2n$. Bernstein proved that if
$M$ is non-zero, then $d \geq n$
\cite{bernsteinModulesRingDifferential1971}.
The module $M$ is said to be {holonomic} if $d = n$ or if $M = 0$.

The holonomicity of $M$ can be tested algorithmically from a Gröbner basis of $S$ by means of \cref{cor:crit}.
Although the precise origin of this criterion is unclear, one implication appears in Zeilberger~\cite{zeilbergerHolonomicSystemsApproach1990},
and the second condition was used as the definition of holonomicity by Kauers~\cite{kauersDFiniteFunctions2023}.
A complete proof is given in~\cite[Theorem 80]{brochetEfficientAlgorithmsCreative2025}.

\begin{theorem}[Folklore]\label{cor:crit}
Let $S$ be a non-trivial submodule of $W_\xx^r$, let $G$ be a Gröbner basis of $S$,
and let $e_1,\dots,e_r$ be the canonical basis of $W_\xx^r$.
The following statements are equivalent:
\begin{enumerate}
\item The module $W_\xx^r/S$ is holonomic.
\item For every subset $A \subseteq \{\xx,\dd\}$ with $|A| = n+1$ and every $i \in \llbracket 1,r \rrbracket$,
the vector space $S \cap \bK\langle A\rangle e_i$ is non-zero.
\item For every subset $A \subseteq \{\xx,\dd\}$ with $|A| = n+1$ and every $i \in \llbracket 1,r \rrbracket$,
there exists $g \in G$ such that $\lm(g) \in \bK\langle A\rangle e_i$.
\end{enumerate}
\end{theorem}
The notation $\llbracket 1,r \rrbracket$ in the above theorem denotes the set $\{1,\dots,r\}$.

\subsection{Singular Locus}
\paragraph{Principal symbol}
Let $S \subset W_{\tt,\xx}(\tt)^r$ be a left submodule, and let $(\boldsymbol{0},\boldsymbol{1})$
denote the weight vector assigning weight $0$ to the variables $x_1,\dots,x_n$
and weight $1$ to the variables $\partial_{x_1},\dots,\partial_{x_n},\partial_{t_1},\dots,\partial_{t_m}$
(see~\cite[Chapter~1.2]{saitoGrobnerDeformationsHypergeometric2000} for a gentle introduction when $m=0$).
Let
\[
P = \sum_{\aalpha,\bbeta,\ggamma,i}
a_{\aalpha,\bbeta,\ggamma,i}(\tt)\,\xx^{\aalpha}\dd_\xx^{\bbeta}\dd_\tt^{\ggamma}e_i
\in W_{\tt,\xx}(\tt)^r,
\]
with $a_{\aalpha,\bbeta,\ggamma,i}(\tt)\in\bK(\tt)$ and set
\[
c
=
\max\left\{
|\bbeta| + |\ggamma|
\;\middle|\;
\exists\,\aalpha,i \text{ such that } a_{\aalpha,\bbeta,\ggamma,i}\neq 0
\right\}.
\]
The {principal symbol} (or {initial form}) of $P$ with respect to
$(\boldsymbol{0},\boldsymbol{1})$ is the element in $\bK(t)[\xx,\xxi,\zzeta]^r$ defined by
\begin{equation}
\ini_{(\boldsymbol{0},\boldsymbol{1})}(P)
=
\sum_{\substack{\aalpha,i \\ |\bbeta|+|\ggamma|=c}}
a_{\aalpha,\bbeta,\ggamma,i}\,\xx^{\aalpha}\xxi^{\bbeta}\zzeta^{\ggamma}e_i,
\end{equation}
where $\xi_1,\dots,\xi_n,\zeta_1,\dots,\zeta_m$ are commuting variables corresponding to
$\partial_{x_1},\dots,\partial_{x_n},\partial_{t_1},\dots,\partial_{t_m}$.

\paragraph{Initial module and characteristic variety}
The {initial module}
\[
\ini_{(\boldsymbol{0},\boldsymbol{1})}(S)
\subset
\bK(t)[\xx,\xxi,\zzeta]^r
\]
is the ideal generated by the initial forms of all elements of $S$.
The {characteristic variety} of the $W_{\tt,\xx}(\tt)$-module
\[
M = W_{\tt,\xx}(\tt)^r/S
\]
is defined as
\[
\charv(M)
=
V\!\left(
\ini_{(\boldsymbol{0},\boldsymbol{1})}(S)
\right)
\;\subset\;
\bK(t)^{2n+m}.
\]

\paragraph{Singular locus}
The singular locus of $S$ is a set-theoretic overestimation 
of the singularities (with respect to $\xx$) of the solutions of $S$. 
It plays a crucial role in the study of Weyl closure, as stated in~\cref{thm:singlocus2}.
Formally, it is defined as the Zariski closure of the image of
\[
\charv(M) \setminus \{\xi_1=\cdots=\xi_n=0, \zeta_1=0,\dots,\zeta_m=0\}
\]
under the coordinate projection
\[
\bK(t)^{2n+m} \to \bK(t)^n,
\qquad
(\xx,\xxi,\zzeta)\mapsto \xx.
\]
Equivalently, it is defined as the zero set of the elimination ideal
\[
\sing(I)
=
V\left(\left(
\ini_{(\boldsymbol{0},\boldsymbol{1})}(S)
:
\langle \xi_1,\dots,\xi_n,\zeta_1,\dots,\zeta_m\rangle^\infty
\right)
\cap
\bK(t)[\xx]\right),
\]
where, for a module $A\subset \bK(t)[\xx,\xxi,\zzeta]^r$, and an ideal $B\subset \bK(t)[\xx,\xxi,\zzeta]$ the {saturation} of $A$
with respect to $B$ is given by
\begin{equation}
A : B^\infty
=
\left\{
a \in \bK(t)[\xx,\xxi,\zzeta]^r
\;\middle|\;
\exists\,k\in\bN
\text{ such that }
B^k a \subset A
\right\}.
\end{equation}

An algorithm for computing the singular locus was proposed by
Oaku~\cite{oakuComputationCharacteristicVariety1994}.

\section{Holonomic \texorpdfstring{$W_{\tt,\xx}(\tt)$}{W(t,x)(t)}-modules}

In view of \cref{cor:crit}, the notion of holonomicity can be extended to
$W_{\xx,\tt}(\tt)$-modules as follows.

\begin{definition}\label{def:holonomicity-crit2}
Let $M=W_{\xx,\tt}(\tt)^r/S$ be a $W_{\xx,\tt}(\tt)$-module and let $e_1,\dots, e_r$ be the canonical basis of $W_{\xx,\tt}(\tt)^r$.
The module $M$ is said to be holonomic if for every subset 
$A\subseteq \{\xx,\partial_\xx,\partial_\tt \}$ with $|A| = n+1$
and every $i\in\llbracket 1,r\rrbracket$,  the intersection $S\cap\mathbb{K}(\tt)\langle A\rangle e_i$ is non-zero.
\end{definition}

When $m=0$, meaning that no variables $t_i$ are present, 
this definition agrees with the standard definition of holonomic $W_\xx$-modules by~\cref{cor:crit},
and when $n=0$, meaning that no variables $x_i$ are present, 
it agrees with the definition of D-finite modules, 
that is, modules that are finite-dimensional over $\bK(\tt)$.

The equivalence between the second and third point of~\cref{cor:crit}, 
which gave an algorithmic criterion for testing holonomicity, naturally extends to this setting.

\begin{lemma}\label{lemma:crit-hol}
Let $S$ be a non-trivial $W_{\xx,\tt}(\tt)$-module and let $G$ be a Gröbner basis of $S$.
The module $M=W_{\xx,\tt}(\tt)^r/S$ is holonomic if and only if for every subset $A\subseteq  \{\xx,\partial_\xx,\partial_\tt \}$ 
with $|A| = n+1$ and every $i\in\llbracket 1,r\rrbracket$, 
there exists $g\in G$ such that $\lm(g) \in \mathbb{K}(\tt)\langle A\rangle e_i$.
\end{lemma}

The classical equivalence between D-finiteness and holonomicity~\cite{kashiwaraHolonomicSystemsLinear1978,takayamaApproachZeroRecognition1992},
also generalizes to the present setting.

\begin{theorem}
Let $S$ be a submodule of $W_{\xx,\tt}(\xx,\tt)^r$. The following statements are equivalent:
\begin{enumerate}
\label{lemma:equivhol}
\item $W_{\xx,\tt}(\xx,\tt)^r/S$ is $D$-finite.
\item $W_{\xx,\tt}(\tt)^r/(S\cap W_{\xx,\tt}(\tt)^r)$ is holonomic as a $W_{\tt,\xx}(\tt)$-module
\item $W_{\xx,\tt}^r/(S\cap W_{\xx,\tt}^r)$ is holonomic as a $W_{\tt,\xx}$-module
\end{enumerate}
\end{theorem}

\begin{proof}
The theorem is proved in the case $r=1$; the general case follows by the same arguments but requires heavier notation.
It is well known that (1) and (3) are equivalent: 
the implication $(1)\Rightarrow(3)$ is proved by 
constructing a suitable filtration~\cite[Appendix]{takayamaApproachZeroRecognition1992},
and the converse is a straightforward application of 
\cref{cor:crit}.
It therefore suffices to prove $(3)\Rightarrow(2)$ and $(2)\Rightarrow(1)$. They are consequences of \cref{def:holonomicity-crit2,cor:crit}.
Only the implication $(3)\Rightarrow(2)$ is detailed; the other one is analogous.

Let $A \subseteq \{\xx,\partial_\xx,\partial_\tt\}$ be a set of cardinality $n+1$.
Then $A\cup\{\tt\}$ has cardinality $n+m+1$.
Applying \cref{cor:crit} to $S\cap W_{\xx,\tt}^r$ and the set $A\cup\{\tt\}$ yields that the $\bK$-vector space
\[
\bK\langle A\cup\{\tt\}\rangle \cap \left(S\cap W_{\xx,\tt}^r\right)
\]
is non-zero.
This implies that the $\bK(\tt)$-vector space
\[
\bK(\tt)\langle A\rangle \cap \left(S\cap W_{\xx,\tt}(\tt)^r\right)
\]
is non-zero, which proves (2).
\end{proof}

\section{Partial Weyl closure by saturation}
The goal of this section is to prove~\cref{thm:singlocus2},
 which states that the partial Weyl closure of a module $S \subset W_{\xx,\tt}(\tt)^r$ 
 can be computed by inverting a single polynomial vanishing on the singular locus of $S$.
A similar result was proved by Tsai~\cite[Theorem~2.2.1]{tsaiAlgorithmsAlgebraicAnalysis2000}
in the case where no parameters $\tt$ is present.
The argument given here is a straightforward extension of his approach to this mixed setting.

\begin{theorem}
    \label{thm:singlocus2}
    Let $S$ be a submodule of $W_{\xx,\tt}(\tt)^r$ with finite rank and let $f\in\bK(\tt)[\xx]$ 
    be a polynomial vanishing on $\sing(S)$. 
    Then, 
    \begin{equation}
     \cl_\xx(S) = W_{\tt,\xx}(\tt)[1/f]S\cap W_{\tt,\xx}(\tt).
    \end{equation}
\end{theorem}

The proof of \cref{thm:singlocus2} relies on \cref{lem:singloc}, which is
a mixed Weyl algebra analogue of a classical result in $D$-module theory
(see, for instance,~\cite[Lemma~6]{grangerBasicCourseDifferential1997}).

\begin{lemma}\label{lem:singloc}
Let $f \in \bK(\tt)[\xx]$ and let $S$ be a submodule of $W_{\xx,\tt}(\tt)[1/f]^r$ such that 
\begin{equation}
M = W_{\xx,\tt}(\tt)[1/f]^r/S
\end{equation}
is finitely generated as a $\bK(\tt)[\xx,1/f]$-module. 
Then $M$ is torsion-free over $\bK(\tt)[\xx,1/f]$.
Equivalently, for any $a \in \bK(\tt)[\xx,1/f]$ and
$g \in W_{\xx,\tt}(\tt)[1/f]^r$, the implication
$a g \in S \Rightarrow g \in S$
holds.
\end{lemma}

\begin{proof}[Proof of \Cref{lem:singloc}]
Let $T$ be the torsion module 
\begin{equation} 
    T = \{ m\in M \mid \exists p\in \bK(\tt)[\xx,1/f], \, pm = 0 \}.
\end{equation}
We want to prove that $T$ is $\{0\}$. 
We start by proving that it is a submodule of $M$.
The stability under left multiplication by $\xx_i$ and $1/f$ is immediate. 
Let $\delta \in \{\dd_\xx, \dd_\tt\}$, let $m \in T$, 
and choose $p \in \bK(\tt)[\xx,1/f]$ such that $p m = 0$.
Multiplying by $\delta$ yields
\begin{equation}
\delta pm = 0 = (\delta\cdot p)m + p (\delta\cdot m)
\end{equation}
Since $\delta\cdot p \in \bK(\tt)[\xx,1/f]$, the element $(\delta\cdot p)m$ lies in $T$.
Hence there exists $q \in \bK(\tt)[\xx,1/f]$ such that
$q (\delta\cdot p)\, m = 0$, which implies $p q\, (\delta\cdot m) = 0$.
Therefore $\delta\cdot m \in T$, proving the claim.

Since $M$ is finitely generated over $\bK(\tt)[\xx,1/f]$, so is $T$.
Let $m_1,\dots,m_\ell$ be generators of $T$, and choose
$p_1,\dots,p_\ell \in \bK(\tt)[\xx]$ such that $p_i m_i = 0$. 
Note that the $p_i$ can indeed be chosen without $f$ in the denominator.
Then the polynomial
\[
P = \prod_{i=1}^\ell p_i
\]
annihilates every element of $T$.
Applying any $\delta \in \{\dd_\xx, \dd_\tt\}$ for each $m\in T$ gives
\[
0 = \delta P m = (\delta\cdot P) m + P (\delta\cdot m).
\]
Since $P (\delta\cdot m) = 0$, it follows that $\delta\cdot P$ annihilates $m$. 
and hence annihilates $T$.
Iterating this argument for suitable choices of $\delta$, we conclude that
$1$ annihilates $T$, which proves that $T = \{0\}$.
\end{proof}

\begin{proof}[Proof of \Cref{thm:singlocus2}]

The inclusion
\[
W_{\tt,\xx}(\tt)[1/f]\,S \cap W_{\tt,\xx}(\tt)
\subseteq
\cl_\xx(S)
\]
is immediate.
We prove the reverse inclusion.
We first prove that for any $i \in \llbracket 1,r \rrbracket$ and any
$\varsigma \in \{\xxi,\zzeta\}$, there exist integers $a,n \in \bN$ such that
\[
f^n \varsigma^a e_i \in \ini_{(\boldsymbol{0},\boldsymbol{1})} (S).
\]
Indeed, since $S$ is a finite-rank module, for any
$i \in \llbracket 1,r \rrbracket$ and any
$\varsigma  \in \{\dd_\xx, \dd_\tt\}$,
there exists a nonzero element
$L' \in W_{\tt,\xx}(\tt,\xx)S\cap \bK(\tt,\xx)\langle \varsigma \rangle e_i$. 
Multiplying it by a suitable polynomial, we obtain a non-zero element 
$L \in  S \cap \bK(\tt)[\xx]\langle\varsigma  \rangle e_i$.
Its initial form $\ini_{(\boldsymbol{0},\boldsymbol{1})}(L)$ then belongs to
$\bK(\tt)[\xx]\langle\varsigma  \rangle e_i$ and is of the form
$g \varsigma ^a e_i$ with $g \in \bK(\tt)[\xx]$.
By definition of $f$, the polynomial $g$ divides a power of $f$.
let $f^n$ denote this power, then
\[
\ini_{(\boldsymbol{0},\boldsymbol{1})}\!\left(\frac{f^n}{g} L\right) = f^n \varsigma ^a e_i,
\]
which proves the claim.
It follows that the quotient
\[
\frac{\bK(\tt)[\xx,\xxi,\zzeta,1/f]^r}{\ini_{(\boldsymbol{0},\boldsymbol{1})}(S)[1/f]}
\]
is finitely generated over $\bK(\tt)[\xx,1/f]$.
Consequently,
\[
N = \frac{W_{\tt,\xx}(\tt)[1/f]^r}{S[1/f]}
\]
is finitely generated as a $\bK(\tt)[\xx,1/f]$-module,
and by \cref{lem:singloc}, the module $N$ is torsion-free.

Finally, let $a \in \cl_\xx(S)$.
There exists $q \in \bK(\tt)[\xx]$ such that $q a \in S$.
Thus the class of $q a$ vanishes in $N$, and torsion-freeness implies
that the class of $a$ also vanishes.
Hence $a \in S[1/f]$.
Since $a \in W_{\tt,\xx}(\tt)^r$, this proves the result.
\end{proof}

\cref{thm:singlocus2} can be reformulated in a more algorithmic way using left saturation module 

The left saturation module of $S\subseteq W_{\xx,\tt}(\tt)^r$ with respect to the polynomial $f\in \bK[\xx,\tt]$ is the left module
\begin{equation}\label{eq:def-sat}
S : (f)^\infty = \left\{ L\in W_{\xx,\tt}(\tt)^r \mid \exists i\in\bN, \, f^iL \in S \right\}.
\end{equation} 
\begin{theorem}
The set $S : (f)^\infty$ is a left submodule of $W_{\xx,\tt}(\tt)^r$
\end{theorem}
\begin{proof}
The set $S : (f)^\infty$ is clearly non-empty, contained in $W_{\xx,\tt}(\tt)^r$ and stable under addition. 
It remains to prove the stability under left multiplication by elements of $W_{\xx,\tt}(\tt)$. 
The stability under multiplication by $\bK(\tt)[\xx]$ is immediate. 
We now check the stability under the action of $\partial_{x_i}$ (the case $\partial_{t_i}$ is analogous).
Let $L\in W_{\xx,\tt}(\tt)$ and $j\in\bN$ such that $f^jL \in S$, 
then a computation shows that
\begin{equation}
f^{j+1} \partial_{x_i}L = \left(\partial_{x_i}f - (j+1)\frac{\partial f}{\partial x_i}\right) f^{j}L.
\end{equation}
The right-hand side belongs to $S$ as $f^{j}L \in S$ by assumption and $S$ is a left-module.
 Hence, $\partial_{x_i}L\in S : (f)^\infty$, which concludes the proof.
\end{proof}

\begin{corollary} 
\label{cor:singlocus2}
    Let $S$ be a submodule of $W_{\xx,\tt}(\tt)^r$ with finite rank and let $f\in\bK(\tt)[\xx]$ 
    be a polynomial vanishing on $\sing(S)$. 
    Then, 
    \begin{equation}
     \cl_\xx(S) = S : (f)^\infty
    \end{equation}
\end{corollary}

\begin{example}
We present a simple example illustrating the computation of a Weyl closure.
Consider the rational function $f(x,y) = 1/(x^2-y^3)$. It is annihilated by the two operators

\begin{equation}
g_1 = \partial_x(x^2-y^3), \quad g_2=\partial_y(x^2-y^3),
\end{equation} 
which generate a finite-rank ideal in $W_{x,y}$.  
These two operators actually generate the annihilator of $f$ in $W_{x,y}(x,y)$.
However, they do not generate a holonomic ideal. 
A third operator annihilating $f$ can be obtained from 
$g_1$ and $g_2$ using a division by $x^2-y^3$:
\begin{equation} 
    3x\partial_x + 2y\partial_y + 1 = \frac{1}{x^2-y^3}(3xg_1 + 2yg_2).
\end{equation}
The ideal generated by $g_1,g_2,g_3$ corresponds to the Weyl closure of $W_{x,y}g_1 + W_{x,y}g_2$,
which is also the annihilator of $f$ in $W_{x,y}$.
\end{example}

\section{The algorithm}
An algorithm is now described for computing a holonomic approximation of the
partial Weyl closure $\cl_\xx(S)$.
Following Rabinowitsch's trick~\cite{rabinowitschHilbertschenNullstellensatz1930},
a new variable $T$ is introduced, which formally represents $1/f$.
The variable $T$ commutes with $\xx, \tt$, and satisfies for all
$\ell \in \{\xx,\tt\}$
\begin{equation}\label{eq:comm-T}
    \partial_\ell T = T\partial_\ell - f_\ell T^2,
\end{equation}
where $f_\ell = \frac{\partial f}{\partial \ell}$.
These relations define a new algebra $W_{\xx,\tt}(\tt)[T]^r$, whose quotient by
\[
    U_f = \sum_{i=1}^r W_{\xx,\tt}(\tt)[T]\,(fT-1)e_i
\]
is isomorphic to $W_{\xx,\tt}(\tt)[1/f]^r$.

As shown in \cref{lem:sat}, left saturation modules can be computed by
eliminating the variable $T$ in a suitable submodule of this algebra.

\begin{lemma}
    \label{lem:sat}
    Let $S$ be a submodule of $W_{\xx,\tt}(\tt)^r$ and $f$ be a polynomial of $\bK(\tt)[\xx]$, then
    \[
    S : (f)^\infty = \left(W_{\xx,\tt}(\tt)[T]S + U_f\right) \cap W_{\xx,\tt}(\tt)^r.
    \]
\end{lemma}

The proof is based on the following lemma, whose immediate consequence is 
\begin{equation}
U_f = \sum_{i=1}^r (fT-1)W_{\xx,\tt}(\tt)[T]e_i
\end{equation}

\begin{lemma}\label{lem:bilatere}
The ideal generated by $fT-1$ is two-sided, that is 
\[
    W_{\xx,\tt}(\tt)[T](fT-1)  = (fT-1)W_{\xx,\tt}(\tt)[T]
\]
\end{lemma}
\begin{proof}
The direct inclusion is a consequence of the equality 
\[
\partial_{\ell} (Tf-1) = (Tf-1)(\partial_\ell - Tf_\ell) 
\]
for $\ell\in\{\xx,\tt\}$ 
and the reverse inclusion follows from the analogous formula 
\[
(Tf-1)\partial_\ell  = (\partial_\ell + Tf_\ell) (Tf-1).
\]
\end{proof}
\cref{lem:sat} can now be proved.
\begin{proof}
To simplify the notation, the result is proved for $r=1$ only, omitting $e_1$. 
In this case, $U_f$ becomes $U_f = W_{\xx,\tt}(\tt)[T](fT-1)$. Let us begin with the direct inclusion. 
Let $L\in S : (f)^\infty$, then there exists $i\in\bN^*$ such that $f^iL \in S$. Hence, $T^if^iL$ 
decomposes as
\[
    T^if^iL = (T^if^i -1)L + L \in W_{\xx,\tt}(\tt)[T]S.
\]
Besides $fT-1$ divides $T^if^i-1$, which implies by \cref{lem:bilatere} that $(T^if^i-1)$ is in $U_f$.
This proves that $L$ is in $W_{\xx,\tt}(\tt)[T]S + U_f$. Because $L$ is in $W_{\xx,\tt}(\tt)$ by definition, 
 it belongs to the intersection. 

Let us now prove the reverse inclusion. Let $L\in (W_{\xx,\tt}(\tt)[T]S + U_f)\cap W_{\xx,\tt}(\tt)$, 
then $L$ decomposes as 
\[ 
L = \sum_{i=1}^p T^ig_i + a
\]
with $p\in\bN$, $g_1,\dots,g_p\in S$ and $a\in U_f$. Then 
\[
f^pL = \sum_{i=1}^p (T^if^i-1)f^{p-i}g_i + \sum_{i=1}^p f^{p-i}g_i + f^pa  = g + \tilde{a}
\]
with $g = \sum_i f^{p-i}g_i \in S$ and $\tilde{a} = f^pa + \sum_{i=1}^p (T^if^i-1)f^{p-i}g_i$. 
Using again that $fT-1$ divides $T^if^i-1$ and \cref{lem:bilatere}, we obtain that $\tilde{a}$ is in $U_f$. 
However, neither $f^pL$ nor $g$ involves the variable $T$,
thus $\tilde{a}$ is in $U_f\cap W_{\xx,\tt}(\tt)$ and must be zero. This proves that $f^pL$ is in $S$.
\end{proof}

\Cref{lem:sat} suggests a strategy for computing the partial Weyl closure using saturation.
 In the commutative case this would be done 
by fixing a monomial order eliminating $T$ and computing a Gröbner basis with respect to this order. 
However, this is not possible in this non-commutative setting as there is no order on $W_{\xx,\tt}(\tt)[T]^r$ eliminating $T$ 
that is compatible with multiplication, in the sense that the product of leading terms is the leading term of the product (recall~\cref{eq:comm-T}). 
This problem is tackled by seeing $W_{\xx,\tt}(\tt)[T]^r$ as an infinite-rank $W_{\xx,\tt}(\tt)$-module 
and performing Gröbner basis computation on truncations of this infinite-rank module. 
This idea of computing Gröbner bases of truncations of an infinite-rank module has already 
been used by Takayama in his integration algorithm~\cite{takayamaAlgorithmConstructingIntegral1990}. 
As in his original algorithm, the termination criterion used here is based on holonomicity and does not
guarantee that the full saturation module is obtained.

A basis of the $W_{\xx,\tt}(\tt)$-module $W_{\xx,\tt}(\tt)[T]^r$ is $\{T^je_i \mid j\in\bN, i\in\llbracket 1,r\rrbracket\}$. 
The basis elements $T^0e_i$ will be identified with $e_i$. The degree of $g\in W_{\xx,\tt}(\tt)[T]^r$ in $T$,
denoted $\deg_T(g)$, is defined as the largest integer $j$ that appears in the decomposition of $g$ with non-zero coefficient in this basis.
Let $G$ be a finite subset of $W_{\xx,\tt}(\tt)[T]^r$ and $\ell$ be the largest degree in $T$ appearing in $G$. 
Then $W_{\xx,\tt}(\tt)G$ is included in the finitely generated $W_{\xx,\tt}(\tt)$-module $\sum_{i=0}^{\ell}\sum_{j=1}^r W_{\xx,\tt}(\tt)T^{i}e_j$ 
in which it is possible to compute Gröbner bases over $W_{\xx,\tt}(\tt)$.
Coupled with \cref{lem:sat} this leads to~\cref{algo:pwc}.

\begin{alg}{Partial Weyl closure}{pwc}
    \textbf{Input:}
    \begin{algovals}
        \item a presentation $S\subseteq W_{\xx,\tt}(\tt)^r$ of a finite rank module
        \item a polynomial $f\in\bK(\tt)[\xx]$ vanishing on $\sing(S)$
        \end{algovals}
    \textbf{Output:} 
    \begin{algovals}
        \item a module $S' \subset\cl_\xx(S)$ such that $W_{\xx,\tt}(\tt)^r/S'$ is holonomic
    \end{algovals}
    \begin{pseudo}
        Fix an order on $ W_{\xx,\tt}(\tt)[T]^r$ eliminating $T$ \\ 
        $H \gets \gens(S)\cup \{(fT-1)e_i \mid i=1,\dots,r\}$\\
        $G \gets H; \quad s\gets 0$ \\ 
        \kw{while} $W_{\xx,\tt}(\tt)^r/ (W_{\xx,\tt}(\tt)G\cap W_{\xx,\tt}(\tt)^r)$ is not holonomic \\+ 
            $G \gets$ GröbnerBasis$\left(\left\{ T^ig \mid i\in\mathbb{N}, \, g\in H, \, \deg_T(T^ig)\leq s\right\}\right)$ \label{line:addG}\\ 
            $s \gets s+1$ \\-
        \kw{return} $G\cap W_{\xx,\tt}(\tt)^r$
    \end{pseudo}
\end{alg}

\begin{theorem}
\Cref{algo:pwc} terminates and returns  generators $G'$ of a $W_{\xx,\tt}(\tt)$-module included in $\cl_\xx(S)$ such that $W_{\xx,\tt}(\tt)^r/W_{\xx,\tt}(\tt)G'$ is holonomic.
\end{theorem}
\begin{proof}
First, note that $\cl_\xx(S)$ coincides with $S : (f)^\infty$ by \cref{thm:singlocus2}, and that it admits a finite Gröbner basis.
By \cref{lem:sat}, this Gröbner basis can be obtained from a finite set of generators of the $W_{\xx,\tt}(\tt)$-module $W_{\xx,\tt}(\tt)[T]S + U_f$.
For $s$ sufficiently large, all these generators will have been taken into account on line~\ref{line:addG}.
Since the closure is holonomic by \cref{lemma:equivhol}, the algorithm must terminate at or before this iteration.
Correctness follows from \cref{lemma:equivhol}, \cref{thm:singlocus2}, and \cref{lem:sat}.\end{proof}

\begin{remark}
    Let $(G_s)_{s\in\bN}$ be the sequence of  Gröbner bases computed 
    on line 5 of~\cref{algo:pwc} and let $(G'_s)_{s\in\bN}$ be the sequence defined by $G'_s = G_s \cap W_{\xx,\tt}(\tt)^r$.
    Two other possible termination criteria for~\cref{algo:pwc} are:
    \begin{itemize}
        \item stopping at the first $s'$ such that $G'_{s'}= G'_{s'+1}$ and \\
        $W_{\xx,\tt}(\tt)^r / W_{\xx,\tt}(\tt)G'_{s'}$ is holonomic,  
        \item stopping at the first $s' + h$ for which $W_{\xx,\tt}(\tt) / W_{\xx,\tt}(\tt)G'_{s'}$ is holonomic, 
        where $h \in \mathbb{N}$ is a prescribed integer.  
    \end{itemize}
    At present, however, the author does not know any criterion that guarantees the output coincides with $\cl_\xx(S)$ without the use of $b$-functions.
\end{remark}

\section{Timings}\label{sec:appl-timings-wc}

This section presents benchmarks for a proof-of-concept implementation~\cite{brochetJuliaPackageMultivariateCreativeTelescoping2025} 
of \cref{algo:pwc}.
The results are compared with two existing implementations of Tsai's algorithm.
The first one is part of the \texttt{Plural} package~\cite{levandovskyyPluralComputerAlgebra2003} of \textsc{Singular}~\cite{deckerSingular440Computer2024}, 
and the second one is included in the \texttt{Dmodules} package~\cite{leykinDmodulesMacaulay2Package} of \textsc{Macaulay2}~\cite{graysonMacaulay2SoftwareSystem}.

Comparing different algorithms implemented by different authors in different programming languages is always delicate. 
This comparison is particularly delicate as both approaches rely on Gröbner basis computations, 
whose performance is known to depend heavily on the underlying implementation.
For this reason, benchmarks for the Gröbner basis implementations used in these three systems are first presented.

All benchmarks were performed on a cluster, where each node is equipped with two Intel(R) Xeon(R) E5-2660 v2 CPUs and 256 GB of RAM. 
Each CPU has 10 cores with hyper-threading (20 threads). 
All computations were run on a single thread when the cluster was not under heavy load, 
although the workload manager may have scheduled multiple threads on the same CPU, making the effective memory available per thread difficult to estimate.
Moreover, small run-to-run variations in execution time were regularly observed, so minor differences in reported timings should not be considered significant.

\subsection{Gröbner bases}\label{sec:gb-timings}
In \texttt{Plural}, Gröbner bases are computed using slimgb~\cite{brickensteinSlimgbGrobnerBases2010}, 
an algorithm specifically designed to limit intermediate expression growth.
\textsc{Macaulay2}'s \texttt{Dmodules} package uses the standard Buchberger algorithm, invoked through the “gbw” command.
To the best of the author's knowledge, users may specify only a weight vector, but not a general monomial order.
The \textsc{Julia} implementation relies on the F4 algorithm~\cite{faugereNewEfficientAlgorithm1999} and uses a tracer~\cite{traversoGrobnerTraceAlgorithms1989} 
to remember the S-pairs reducing to zero in subsequent calls. 
When working over a finite field, the timing of the first call (learn) and of a subsequent call (apply) is given.
All three implementations are tested over the finite field $\bF_{536870909}$ as well as over $\bQ$.

\paragraph{Gröbner bases over a finite field}
The performance of the three implementations are compared in~\cref{tab:gbmodp} over the field $\bF_{536870909}$ on the following set of examples.
8reg corresponds to the $W_\xx$-module $S$ described in~\cite[Theorem 33]{brochetFasterMultivariateIntegration2025} where $t$ was evaluated at $42$.
Beukers corresponds to the computation of the annihilator of $e^{f}$ where 
\begin{equation}
f= x_4^6 -(x_4^2-x_1x_2)x_3x_4^3 -4049x_1x_2x_3(x_4-x_1)(x_4-x_2)(x_4-x_3).
\end{equation}
This polynomial evaluated at $x_4=1$ appeared in the work of Beukers~\cite{beukersIrrationalityP2Periods1983} and is related to the Apéry numbers.
Exp1 corresponds to the computation of the annihilator of $e^{g}$ where 
\begin{equation}
g=(x_1^2 + x_2^2 + x_3^2)(x_1^4+x_2^4+x_3^4).
\end{equation}
For the above examples, the label (gl) in~\cref{tab:gbmodp} indicates that the order 
\begin{equation}
\grevlex(x_1,\dots,x_\ell,\partial_1,\dots,\partial_\ell)
\end{equation} was 
used in \textsc{Julia} and \textsc{Singular}, while the weight vector $(\mathbf{1},\mathbf{1})$ was chosen in \textsc{Macaulay2}. 
Similarly, the label (elim) indicates that the block order 
\begin{equation}
\lex(x_1,\dots,x_\ell) > \lex(\partial_1,\dots,\partial_\ell)
\end{equation} was used
in \textsc{Julia} and \textsc{Singular}, while the weight vector $(\bm{1},\bm{0})$ was chosen in \textsc{Macaulay2}. 
The example ha2 corresponds to the first Gröbner basis computation in Oaku's algorithm for computing the $b$-function~\cite{oakuAlgorithmComputingBfunctions1997} 
on input the polynomial
\begin{gather}
x_1x_2x_3x_4(x_1 + x_2)(x_1 + x_3)(x_1 + x_4)
\end{gather}
with the order 
\begin{equation}
\text{grevlex}(u,v) > \text{grevlex}(s,t,x_1,\dots,x_\ell,\partial_t,\partial_1,\dots,\partial_\ell).
\end{equation}
This last example is taken from~\cite{levandovskyyComputationalDmoduleTheory2008}.
A star (*) after the timing of a computation with \textsc{Macaulay2} indicates that the returned Gröbner 
basis is different than the other two due to the use of a different monomial order.
MEM indicates that the computation ran out of memory.

\begin{table}

\centering 
\caption{Comparison of three Gröbner basis implementations over the field $\bF_{536870909}$.}
\label{tab:gbmodp}
\begin{tabular}{c c c c c}
\toprule
 & \multicolumn{2}{c}{Julia} & \multicolumn{1}{c}{Macaulay2} & \multicolumn{1}{c}{Singular} \\
\cmidrule(r){2-3} \cmidrule(r){4-4} \cmidrule(r){5-5}
Ideal & learn & apply &  &  \\
\midrule
8reg (elim)
& {\color{mygreen}<1s}
& {\color{mygreen}<1s}
& {\color{myyellow}1s}
& {\color{myyellow}4s} \\

8reg (gl)
& {\color{mygreen}45s}
& {\color{mygreen}4s}
& {\color{myyellow}1h 38m 10s}
& {\color{myyellow}5m 8s} \\

Beukers (elim)
& {\color{mygreen}8s}
& {\color{mygreen}5s}
& {\color{myyellow}20h 26m 49s*}
& {\color{myyellow}31s} \\

Beukers (gl)
& {\color{mygreen}105s}
& {\color{mygreen}30s}
& {\color{myyellow}1h 38m 15s}
& {\color{myyellow}31m 27s} \\

Exp1 (elim)
& {\color{myyellow}3s}
& {\color{myyellow}2s}
& {\color{mygreen}<1s}
& {\color{myyellow}10s} \\

Exp1 (gl)
& {\color{mygreen}<1s}
& {\color{mygreen}<1s}
& {\color{mygreen}<1s*}
& {\color{mygreen}<1s} \\

ha2
& {\color{mygreen}35m 52s}
& {\color{mygreen}31m 53s}
& {\color{myyellow}1h 49m 35s*}
& {\color{myyellow}1h 43m 56s} \\
\bottomrule
\end{tabular}

\end{table}

As shown in~\cref{tab:gbmodp}, the \textsc{Julia} implementation performs significantly better, 
even without the use of a tracer, 
than the implementations in \textsc{Macaulay2} and \textsc{Singular}. 
Moreover, for computations lasting more than a few seconds, 
\textsc{Singular}'s \texttt{slimgb} algorithm substantially outperforms the implementation in \textsc{Macaulay2}. 
The large discrepancy observed for Beukers (elim) can be explained by the fact that 
\textsc{Macaulay2} computes a Gröbner basis that differs from the one produced by the other two systems.

\paragraph{Gröbner bases over $\bQ$}
The same set of examples as in the previous paragraph is used in~\cref{tab:gbQ} to compare the three implementations over $\bQ$.
\textsc{Singular} and \textsc{Macaulay2} perform the computation directly over the rationals,
 while the \textsc{Julia} implementation proceeds modulo several primes and then reconstructs the result in $\bQ$.
The \textsc{Julia} implementation is probabilistic, since no effective bounds are known 
on the size of the coefficients of the output. 
Consequently, a probabilistic reconstruction criterion is used: 
the algorithm terminates when two successive reconstructions succeed 
and yield identical results. 
This criterion is standard in computer algebra and is employed in many software packages.
Somewhat surprisingly, the modular method performs worse than the direct computation, 
even though the \textsc{Julia} implementation outperforms the others over $\bF_{536870909}$.
The efficient computation of Gröbner bases in Weyl algebras would benefit from further investigation,
which is beyond the scope of this article.
Note, however, that Gröbner basis computation is typically only a means to an end, and the reconstruction step can often be postponed or avoided altogether.

\begin{table}
\centering 
\caption{Comparison of three Gröbner basis implementations over the field $\bQ$.}
\label{tab:gbQ}

\begin{tabular}{c c c c}
\toprule
 & Julia & Macaulay2 & Singular \\
\midrule
8reg (elim) &
{\color{mygreen}1s} &
{\color{mygreen}1s} &
{\color{myyellow}21s} \\

8reg (gl) &
{\color{myyellow}22m 44s} &
{\color{myyellow}1h 45m 57s} &
{\color{mygreen}14m 43s} \\

Beukers (elim) &
{\color{myyellow}2m 20s} &
{\color{myred}MEM} &
{\color{mygreen}39s} \\

Beukers (gl) &
{\color{myyellow}2h 59m 53s} &
{\color{myyellow}2h 19m 57s} &
{\color{mygreen}1h 29m 33s} \\

Exp1 (elim) &
{\color{myyellow}12s} &
{\color{mygreen}<1s*} &
{\color{myyellow}12s} \\

Exp1 (gl) &
{\color{mygreen}<1s} &
{\color{mygreen}<1s} &
{\color{mygreen}<1s} \\

ha2 &
{\color{myyellow}5h 7m 13s} &
{\color{mygreen}1h 48m 32s*} &
{\color{myyellow}1h 48m 42s} \\
\bottomrule
\end{tabular}
\end{table}

\subsection{Weyl closure}
The three implementations of the Weyl closure are now compared in~\cref{tab:wc}.
Recall that~\cref{algo:pwc} only computes an approximation of the Weyl closure, while Tsai's algorithm is exact. 
Accordingly, we report the size of the reduced Gröbner basis obtained at the end of the computation, 
except for the \textsc{Macaulay2} implementation, where a different monomial order is used.
The benchmark set consists of computing the Weyl closure of the annihilating ideals of the following functions, 
as returned by the \textsc{Mgfun} package~\cite{chyzakMgfunMaplePackage}.
\begin{equation}
\begin{aligned}
\makebox[\linewidth][l]{$f_1(a,x) =
x\,\J_1(ax)\,\I_1(ax)\,\Y_0(x)\,\K_0(x)
$}
\\
\makebox[\linewidth][r]{\text{grevlex}$(a,x,\partial_a,\partial_x)$}
\end{aligned}
\label{eq:4Bessels}
\end{equation}

\begin{equation}
\begin{aligned}
\makebox[\linewidth][l]{$f_2(x,y,z,t)=
\frac{1}{1-(1-xy)z - txyz(1-x)(1-y)(1-z)}
$}
\\
\makebox[\linewidth][r]{\text{grevlex}$(t,x,y,z,\partial_t,\partial_x,\partial_y,\partial_z)$}
\end{aligned}
\label{eq:beukers-tim}
\end{equation}

\begin{equation}
\begin{aligned}
\makebox[\linewidth][l]{$f_3(t,x,y)=
\J_1(x)\,\J_1(y)\,\J_2\!\bigl(t\sqrt{xy}\bigr)\,e^{-x-y}
$}
\\
\makebox[\linewidth][r]{\text{grevlex}$(t,x,y,\partial_t,\partial_x,\partial_y)$}
\end{aligned}
\label{eq:doubleBessel}
\end{equation}

\begin{equation}
\begin{aligned}
\makebox[\linewidth][l]{$f_4(t,x,y)=
\J_1(x)\,\J_1(y)\,\J_2\!\bigl(t\sqrt{xy}\bigr)\,e^{-x-y}\exp\left(\frac{1}{xyt}\right)
$}
\\
\makebox[\linewidth][r]{\text{grevlex}$(t,x,y,\partial_t,\partial_x,\partial_y)$}
\end{aligned}
\label{eq:doubleBessel_exp}
\end{equation}

\begin{equation}
\begin{aligned}
\makebox[\linewidth][l]{$f_5(t,x,y)=
\J_1(x)\,\J_1(y)\,\J_2\!\bigl(t\sqrt{xy}\bigr)e^{-x-y}\log(xyt)
$}
\\
\makebox[\linewidth][r]{\text{grevlex}$(t,x,y,\partial_t,\partial_x,\partial_y)$}
\end{aligned}
\label{eq:doubleBessel_log}
\end{equation}

\begin{equation}
\begin{aligned}
\makebox[\linewidth][l]{$f_6(t,x,y)=
\frac{1}{(y^2+1)(x^2+1)\,t - xy}
$}
\\
\makebox[\linewidth][r]{\text{grevlex}$(t,x,y,\partial_t,\partial_x,\partial_y)$}
\end{aligned}
\label{eq:ssw2}
\end{equation}

\begin{equation}
\begin{aligned}
\makebox[\linewidth][l]{$f_7(t,x,y)=
\frac{\exp\left(\frac{1}{(y^2+1)(x^2+1)\,t - xy}\right)}{(y^2+1)(x^2+1)\,t - xy}
$}
\\
\makebox[\linewidth][r]{\text{grevlex}$(t,x,y,\partial_t,\partial_x,\partial_y)$}
\end{aligned}
\label{eq:ssw2_exp}
\end{equation}

\begin{equation}
\begin{aligned}
\makebox[\linewidth][l]{$f_8(t,x,y)=
\frac{\log((y^2+1)(x^2+1)\,t - xy)}{(y^2+1)(x^2+1)\,t - xy}
$}
\\
\makebox[\linewidth][r]{\text{grevlex}$(t,x,y,\partial_t,\partial_x,\partial_y)$}
\end{aligned}
\label{eq:ssw2_log}
\end{equation}

\begin{equation}
\begin{aligned}
\makebox[\linewidth][l]{$f_9(x,u)=
\Ai \bigl(2^{2/3}(x^2+u)\bigr)
$}
\\
\makebox[\linewidth][r]{\text{grevlex}$(x,u,\partial_x,\partial_u)$}
\end{aligned}
\label{eq:airy}
\end{equation}

\begin{equation}
\begin{aligned}
\makebox[\linewidth][l]{$f_{10}(t,x,y,z,w)=
\frac{1}{q(x,y,z,w)t - xyzw}
$}
\\
\makebox[\linewidth][r]{\text{grevlex}$(t,x,y,z,w,\partial_t,\partial_x,\partial_y,\partial_z,\partial_w)$}
\end{aligned}
\label{eq:CY}
\end{equation}

where in the last equation 
\begin{equation*}
       \begin{aligned}
        q ={}& x^2y^2z^2w^2 + x^2y^2zw^2 + x^2yz^2w^2 + xy^2zw^2 + xy^2w^2 + xyzw^2 \\
        &+ xyz + xz^2 + xy + xz + y + z .
    \end{aligned}
\end{equation*}
The symbols  $\J, \I,\Y,\K$ denote the usual Bessel functions and $\Ai$ denotes the Airy function of the first kind.
\cref{eq:4Bessels,eq:doubleBessel} were taken from~\cite{chyzakABCCreativeTelescoping2014},
\cref{eq:beukers-tim} appeared in~\cite{beukersIrrationalityP2Periods1983},
\cref{eq:ssw2} belongs to a family of integrals related to the enumeration of small step walks in the quarter plane~\cite{bostanHypergeometricExpressionsGenerating2017} (with a polynomial factor in the numerator removed),
 \cref{eq:airy} is from the database~\cite{koutschanHolonomicFunctionsMathematica2014}, 
 and~\cref{eq:CY} is the period of a Calabi--Yau 3-fold appearing in the database
 associated with the article~\cite{batyrevConstructingNewCalabiYau2010} (Topology \#42, polytope v25.59).
The functions \cref{eq:doubleBessel_exp,eq:doubleBessel_log,eq:ssw2_log,eq:ssw2_exp}
are obtained from \cref{eq:doubleBessel,eq:ssw2} by artificially increasing
the severity of their singularities.
This is achieved by
multiplying the original functions by a factor of the form
$\exp(1/u)$ or $\log(u)$, where $u$ is a polynomial vanishing on their
corresponding singular locus.
Note that this benchmark table does not include any case of partial Weyl closure since there are no algorithms to compare against.
\begin{table}
\centering 
\caption{Comparison of three algorithms for the Weyl closure. 
The column time indicates the computation time and the column size the number of operators in the output Gröbner basis.}
\label{tab:wc}
\begin{tabular}{c c c c c c}
\toprule 
& \multicolumn{2}{c}{Julia} & \multicolumn{1}{c}{Macaulay2} & \multicolumn{2}{c}{Singular} \\ 
\cmidrule(r){2-3} \cmidrule(r){4-4} \cmidrule(r){5-6}
  & time & size & time & time & size \\  

\cref{eq:4Bessels} &
\textcolor{myyellow}{198} & 19 &
\textcolor{myyellow}{13.9} &
\textcolor{mygreen}{8} & 32 \\

\cref{eq:beukers-tim} &
\textcolor{mygreen}{0.9} & 37 &
\textcolor{myred}{MEM} &
\textcolor{myred}{MEM} & - \\

\cref{eq:doubleBessel} &
\textcolor{mygreen}{0.4} & 8 &
\textcolor{myyellow}{53.5} &
\textcolor{myyellow}{440} & 13 \\

\cref{eq:doubleBessel_exp} &
\textcolor{mygreen}{82} & 84 &
\textcolor{myred}{MEM} &
\textcolor{myred}{MEM} & - \\

\cref{eq:doubleBessel_log} &
\textcolor{myred}{>48h} & - &
\textcolor{myred}{MEM} &
\textcolor{myred}{MEM} & - \\

\cref{eq:ssw2} &
\textcolor{mygreen}{0.05} & 13 &
\textcolor{myred}{MEM} &
\textcolor{myyellow}{3} & 13 \\

\cref{eq:ssw2_exp} &
\textcolor{mygreen}{2} & 28 &
\textcolor{myred}{MEM} &
\textcolor{myyellow}{34} & 28 \\

\cref{eq:ssw2_log} &
\textcolor{mygreen}{58} & 28 &
\textcolor{myyellow}{2173.5} &
\textcolor{myred}{MEM} & - \\

\cref{eq:airy} &
\textcolor{mygreen}{0.03} & 4 &
\textcolor{myyellow}{0.2} &
\textcolor{mygreen}{<1} & 4 \\

\cref{eq:CY} &
\textcolor{myred}{>48h} & - &
\textcolor{myred}{MEM} &
\textcolor{myred}{MEM} & - \\

\end{tabular}

\end{table}

We observe that \textsc{Julia} outperforms the other two algorithms on all examples except \cref{eq:4Bessels}.
However, in both examples involving Bessel functions, \textsc{Julia} does not compute the full Weyl closure.
Note that the modular method, which does not appear to be very fruitful for Gröbner bases over $\bQ$, 
becomes more relevant in this setting, since operators in the localization $W_{\xx}[1/p]$ typically have much larger coefficients than those in the closure $\cl(S)$.

Lastly, the holonomic annihilator of rational functions can be obtained by dedicated algorithms~\cite{ oakuAlgorithmComputingBfunctions1997, brianconRemarquesLidealBernstein2002, saitoGrobnerDeformationsHypergeometric2000,andresPrincipalIntersectionBernsteinsato2009}. 
For example, the annihilator of the integrand in~\cref{eq:beukers-tim} can be obtained in roughly 30 seconds in \textsc{Singular}.

\bibliographystyle{ACM-Reference-Format}
\bibliography{biblio}

\end{document}